\documentclass[aps,pra,showpacs,twoside,twocolumn,longbibliography,10pt,nofootinbib]{revtex4-2}
\usepackage[colorlinks=true, citecolor=blue, urlcolor=blue, linkcolor = blue ]{hyperref}
\usepackage{epsfig,newlfont,amssymb,amsfonts,amsmath,bm,subfigure,mathtools,amsthm,braket,times,soul,enumitem,color}
\usepackage[normalem]{ulem}
\newcommand{\stkout}[1]{\ifmmode\text{\sout{\ensuremath{#1}}}\else\sout{#1}\fi}
\usepackage[english]{babel}
\usepackage[utf8]{inputenc}
\usepackage{array}
\usepackage{xcolor}
\usepackage{graphics}

\newtheorem{proposition}{Proposition}

\newtheorem{lemma}{Lemma}

\newcommand{\ketbra}[2]{|#1\rangle \langle #2|}
\def\Tr{\text{Tr}}

\usepackage{amsthm}
\usepackage{verbatim}
\usepackage{bbm}
\usepackage{wrapfig}

\usepackage{hyphenat}

\newlength\figureheight 
\newlength\figurewidth 
\begin{document}

\title{Surpassing thermal-state limit in thermometry via non-completely positive quantum encoding} 


\author{Anindita Sarkar}

\author{Paranjoy Chaki$^{*}$} 

\author{Debarupa Saha$^{*}$} 
\author{Ujjwal Sen}

\affiliation{Harish-Chandra Research Institute, Chhatnag Road, Jhunsi, Prayagraj  211 019, India}
 \affiliation{Homi Bhabha National Institute, Training School Complex, Anushakti Nagar, Mumbai 400 094, India}

 \thanks{These authors contributed equally to this work.}

\begin{abstract}
Conventional quantum thermometry assumes completely positive (CP) encoding maps, where the probe is initially uncorrelated with the environment. We consider realistic scenarios with initial probe–environment correlations leading to physically realizable non-completely positive (NCP) encoding, and show how such encodings can significantly impact temperature estimation of the environment. We first consider pure entangled probe-environment initial states (Type-I NCP encoding) and analytically show that for probes and environments of equal but arbitrary dimension, the maximum achievable precision matches the thermal-state bound, as in the CP case. However, upon relaxing the constraint of pure probe-environment states and considering general correlated initial states (Type-II NCP encoding), we demonstrate that the estimation precision can surpass the thermal-state limit. This establishes a clear advantage of NCP encoding in enhancing thermometric performance. We illustrate the results using 
qubit probes interacting with qubit environments via XY interactions.

\end{abstract}
\maketitle
\section{Introduction}

Quantum thermometry~\cite{PhysRevA.82.011611,PhysRevA.84.032105,PhysRevA.91.012331,Correa2015-thermometry,PhysRevA.92.032112,Paris2016-thermometry,Xie2017,Campbell_2017,PhysRevA.96.062103,DePasquale2018-thermometry,PhysRevA.97.032129,PhysRevA.98.042124,PhysRevA.98.050101,Razavian2019,Mehboudi2019-thermometry,Potts2019-thermometry,PhysRevLett.123.180602,PhysRevA.101.032112,PhysRevA.102.012204,Jorgensen2020-thermometry,PhysRevA.102.042417,Mok2021,Hovhannisyan2021-thermometry,PhysRevResearch.3.043039,Rubio2021-thermometry,PhysRevLett.128.040502,Cenni2022thermometryof,Aybar2022criticalquantum,PRXQuantum.3.040330,PhysRevA.107.042614,PhysRevA.108.032220,PhysRevResearch.5.043184,zhang2024achievingheisenbergscalinglowtemperature,Sone_2024,PhysRevA.110.032605,PhysRevA.111.052216,t74h-c8kw,gsh7-r7ms,PhysRevA.102.042417,62ks-19fs,PhysRevResearch.5.043184,PhysRevA.107.042614,prs2-jcxj,xie2026directtemperaturereadoutnonequilibrium,saha2026precisionenhancementtransientquantum} 
deals with estimating the temperature of a quantum system (the environment), which is typically not directly accessible, by employing an auxiliary quantum system (the probe) on which measurements can be performed. 
It plays a central role in the characterization of quantum thermodynamic processes
~\cite{RevModPhys.81.1}, in the control of quantum thermal devices and also have applications in quantum simulation~\cite{RevModPhys.80.885,RevModPhys.86.153}. Experimentally, quantum thermometry has been realized across a variety of physical platforms~\cite{doi:10.1073/pnas.1306825110,Neumann2013,10.1063/1.4823703,10.1063/5.0044824,PhysRevApplied.23.054079}.

A conventional thermometry protocol consists of three fundamental steps: preparation of an appropriate probe state, encoding of the temperature of the environment into the probe via an interaction process, and subsequent measurement of the evolved probe state. 
Temperature estimation is most commonly performed within the framework of quantum metrology~\cite{Helstrom1969,HOLEVO1973337,PhysRevLett.72.3439,BRAUNSTEIN1996135,doi:10.1126/science.1104149,PhysRevLett.94.020502,PhysRevLett.96.010401,doi:10.1142/S0219749909004839,Giovannetti2011,holevo,RevModPhys.89.035002,RevModPhys.90.035005,RevModPhys.90.035006,Mukhopadhyay_2025,doi:10.1142/S0129183125430065}, although alternative approaches also exist~\cite{PhysRevLett.119.090603,ZHANG2019113635,HUANGFU2021127172}. 
Within this framework, the estimation precision is quantified by the quantum Fisher information (QFI)~\cite{PhysRevLett.72.3439,BRAUNSTEIN1996135}, which sets the ultimate precision bound through the quantum Cram\'er-Rao inequality~\cite{crap1,PhysRevD.23.357,PhysRevLett.72.3439,holevo}. 
A larger value of the QFI therefore corresponds to a higher achievable precision and serves as a fundamental figure of merit for thermometric protocols.

Most studies in quantum thermometry assume that the probe is initially uncorrelated with the environment prior to the encoding step, such that the resulting dynamics are described by a completely positive trace-preserving (CPTP) map~\cite{k,Alicki_Lendi,preskill1998}. 
Such dynamics arise from the application of a global unitary on the product probe--environment state, followed by tracing out the environment.

However, in realistic scenarios, the probe may be initially correlated with the environment. In such cases, the reduced dynamics of the probe need not be completely positive, and the corresponding encoding processes are described by non-completely positive trace-preserving (NCPTP) maps~\cite{PhysRevLett.73.1060,PhysRevA.64.062106,PhysRevA.70.052110,ziman2006quantumprocesstomographyrole,PhysRevA.77.042113,R2008,wood2009noncompletelypositivemapsproperties}. 
Non-completely positive maps and non-positive measurements have been shown to be advantageous in various quantum information tasks, including entanglement detection~\cite{HORODECKI19961}, state discrimination~\cite{Arai_2019,PhysRevLett.125.150402,choudhary2026enhancedquantumstatediscrimination}, quantum metrology~\cite{chaki2025nonpositivemeasurementsarentbeneficial}, discrimination-based quantum games~\cite{Regula2021operational}, and energy extraction from quantum batteries~\cite{PhysRevLett.132.240401,chaki2025positivenonpositivemeasurementsenergy}. 
Motivated by these observations, we investigate whether NCPTP encodings can enhance thermometric precision beyond that achievable from the thermal state of the environment.

To this end, we consider two classes of physically realizable NCPTP encodings: type-I and type-II, corresponding to pure entangled and general correlated initial probe-environment states, respectively. 
In both cases, the marginal state of the environment is fixed to be a thermal state at the temperature to be estimated, ensuring a fair comparison with CPTP encoding.

We show analytically that, for probes and environments of equal (but otherwise arbitrary) dimension, the maximum QFI achievable under type-I NCPTP encoding does not exceed that of the thermal state of the environment, a scenario identical to that of CPTP encoding.
However, for a qubit probe interacting with a qubit environment, we demonstrate numerically that type-II NCPTP encoding can surpass this bound when optimized over general two-qubit initial states and global unitaries.
The enhancement is observed for both arbitrary two-qubit unitaries and physically relevant energy-conserving interactions. This indicates that, although CPTP protocols already achieve the maximum precision attainable within the type-I setting, the additional freedom provided by general correlated initial states in type-II encoding can lead to a genuine precision advantage.

Our results thus identify scenarios in which realistic initial correlations can be harnessed to enhance thermometric performance. 
In particular, unlike previous works~\cite{Sekatski2022optimal,PhysRevA.108.062421,e26070568,PhysRevE.110.024132,trombetti2025nonequilibriumthermometryensembleinitially,prs2-jcxj,saha2026precisionenhancementtransientquantum}, where non-equilibrium enhancement arises within the CPTP framework relative to the probe's thermal state, we demonstrate an enhancement relative to the thermal-state precision of the environment itself. 
This highlights the potential of NCPTP encodings in the design of more sensitive quantum thermometers.

We note that the role of initial probe-environment correlations in thermometry has also been investigated in Refs.~\cite{PhysRevA.104.012211,Zhang_2024,Mirza_2024}. 
In those works, correlations are first consumed via a measurement-based preparation step, resulting in an effective CPTP encoding. 
In contrast, our approach directly exploits the correlated probe-environment state during the encoding process, leading to genuinely NCPTP dynamics.

The rest of the paper is organized as follows. Sec.~\ref{prelim} is devoted to a detailed discussion of quantum thermometry within the framework of parameter estimation theory, along with a brief outline of non-completely positive trace-preserving maps. In Sec.~\ref{th-gen}, we study the problem of temperature estimation of a thermal environment, using CPTP encoding, as well as type-I and type-II NCPTP encodings. We analytically establish the existence of an upper bound on the optimal precision for the CPTP and type-I NCPTP encodings, given by the quantum Fisher information of the thermal environment. Next, we numerically demonstrate, for qubit probe and environment, that the type-II NCPTP encoding can surpass this bound by optimizing the quantum Fisher information over general two-qubit correlated initial states, considering arbitrary two-qubit unitaries in Sec.~\ref{gen} and energy-preserving unitaries in Sec.~\ref{en}. We further show that the maximum quantum Fisher information, obtained by optimizing over initial probe states in the CPTP encoding and over single-qubit local unitaries acting on the probe in the type-I NCPTP encoding, attains the thermal-state value in both cases. We present the conclusion in Sec.~\ref{conc}.
\section{Preliminaries}\label{prelim}
In this section, we cover the essential background needed in this work.  Subsection~\ref{qther} outlines the basics of quantum thermometry, while subsection~\ref{ncptp} contains a brief description of NCPTP maps.  
\subsection{Quantum Thermometry}\label{qther}


We elucidate the basic scheme of temperature estimation, using the techniques of quantum parameter estimation theory. We consider a sample maintained at temperature \(T\). To estimate its temperature, an auxiliary probe is allowed to interact with the sample, thereby encoding information about \(T\) in the probe state. Because temperature is not an observable in quantum mechanics, its value is inferred from the temperature-dependent measurement statistics of a suitably chosen observable.

Let the encoded probe state be denoted by $\rho_S(T)$. In quantum mechanics, any measurement is represented by a set of operators $\{\Pi_\theta\}$, satisfying $\Pi_\theta\geq 0$ and $\sum_{\theta} \Pi_\theta=\mathbb{I}_d$. Here, $\mathbb{I}_d$ denotes the identity operator on a $d$-dimensional space. The values $\theta$ in the subscript denote the possible measurement outcomes obtained, with each outcome $\theta$ associated with a probability $p_{T}(\theta)\coloneqq\Tr(\Pi_\theta \rho_S(T))$, given by Born rule~\cite{Born1926}. 
Once the measurement outcomes $\{\theta\}$ and the associated probability distribution $p_{T}(\theta)$ have been obtained, an estimator function $\hat{T}(\theta)$ is employed to obtain an estimate of the temperature from the data. Here we confine ourselves to estimators satisfying the condition of local unbiasedness at $T=T_0$: 
\begin{align}
   \int d\theta \hat{T}(\theta) p_{T_0}(\theta)=T_0, \; \int d\theta \hat{T}(\theta) \frac{dp_{T_0}(\theta)}{dT}\bigg \vert_{T_0}=1, \label{lub}
\end{align}
where $T_0$ represents the actual value of the temperature. Such estimators are referred to as locally unbiased estimators. 

The objective of quantum parameter estimation theory is to estimate the parameter of interest (here, the temperature) with the highest possible precision. This is achieved by minimizing the variance corresponding to all unbiased estimators and over all possible measurement settings. For a fixed probe state, the optimal precision of any unbiased estimator is then constrained by the quantum Cram\'er-Rao bound.~\cite{PhysRevD.23.357,PhysRevD.23.357,PhysRevLett.72.3439,holevo},
\begin{align}
    \Delta^2 \hat{T}\geq\frac{1}{NF_Q},\label{qcrb}
\end{align}
where $F_Q$ is the quantum Fisher information (QFI)~\cite{PhysRevLett.72.3439,BRAUNSTEIN1996135,doi:10.1142/S0219749909004839,Liu_2020}, defined by~\cite{Liu_2020}
\begin{equation}
F_Q(\rho_S(T))=2\sum_{\substack{i,j=0 \\ \lambda_i+\lambda_j \neq 0}}^{d-1}\frac{|\bra{\lambda_i}\frac{d\rho_S(T)}{dT}\ket{\lambda_j}|^2}{\lambda_i+\lambda_j}.
    \label{qfi}
\end{equation}
In our work, we consider single-shot scenario i.e., $N=1$.
\subsection{Quantum Maps for Encoding}\label{ncptp}

Let $\mathcal{D}(\mathcal{H}^d)$ denote the set of all positive semi-definite, unit-trace linear operators acting on a $d$-dimensional Hilbert space $\mathcal{H}^d$. Each element of $\mathcal{D}(\mathcal{H}^d)$ represents a quantum state.

 A quantum map $\mathcal{N}: \mathcal{D}(\mathcal{H}^d) \rightarrow\mathcal{D}(\mathcal{H}^{d'})$ takes a quantum state $\rho_S \in \mathcal{D}(\mathcal{H}^d)$ belonging to a system $S$ as input and produces
a quantum state $\rho'_{S'} \in \mathcal{D}(\mathcal{H}^{d'})$ corresponding to another system $S'$ as output, i.e., $\rho'_{S'} \coloneqq \mathcal{N} (\rho_S)$. From the definition of $\mathcal{N}$, it follows that $\mathcal{N}(\rho_S) \geq 0$ and  $\Tr[\mathcal{N}(\rho_S)]=\Tr(\rho_S)$, therefore $\mathcal{N}$ is positive and trace-preserving. Furthermore,  $\mathcal{N}$ is completely positive and trace-preserving (CPTP) if the extended map $\mathcal{N} \otimes \mathbbm{I}_{R} \geq 0$ for any $R \in \mathbb{N}$. Here, $\mathbbm{I}_{R}$ is the identity operator on an $R$-dimensional Hilbert space. 

Any completely positive trace-preserving map acting on a system \(S\) initially in the state \(\rho_S\) can be realized by introducing an external auxiliary (or environment) system \(E\) prepared in a state \(\rho_E\). Assuming that the system and the environment are initially uncorrelated, i.e., in the product state \(\rho_S \otimes \rho_E\), a global unitary operation \(U_{SE}\) is applied to the composite system. Tracing out the environmental degrees of freedom then induces a CPTP map~\cite{k,Alicki_Lendi} on the system, given by
\begin{equation}
\mathcal{N}_{\mathrm{CPTP}}(\rho_S)
:= \operatorname{Tr}_E \!\left[ U_{SE} \left( \rho_S \otimes \rho_E \right) U_{SE}^\dagger \right]. \label{qmap-cptp}
\end{equation}

In a general and realistic scenario, the system may carry some correlations with the environment. Suppose that the system interacts with an environment $E$ initially, with the initial joint state given by $\rho_{SE}$. The time evolution of the joint system $SE$ is described by a global unitary operator $U_{SE}$. The time evolution of $S$ alone is given by 
\begin{equation}  \mathcal{N}_\mathrm{G}(\rho_S)\coloneqq \Tr_E [U_{SE} \rho_{SE} U_{SE}^\dagger],
\label{qmap-gen}
\end{equation}
 The quantum correlations, including entanglement between the system and the environment can give rise to non-completely positive trace-preserving (NCPTP) maps~\cite{PhysRevLett.73.1060,PhysRevA.64.062106,PhysRevA.70.052110,ziman2006quantumprocesstomographyrole,PhysRevA.77.042113,R2008,wood2009noncompletelypositivemapsproperties}. Therefore, the set of general quantum maps \(\mathrm{G}\) can be written as \(\mathrm{G} = \mathrm{CPTP} \cup \mathrm{NCPTP}\), where the sets \(\mathrm{CPTP}\) and \(\mathrm{NCPTP}\) are disjoint.

\section{Present Framework for Quantum Thermometry}\label{th-gen}

In this section, we introduce our protocol for estimating the temperature \(T\) of an arbitrary-dimensional thermal environment \(E\) which is initially in a thermal state \(\tau_E \coloneqq \frac{\exp(-H_E/T)}{\Tr[\exp(-H_E/T)]}\). Here $H_E$ is the Hamiltonian of the environment. The temperature is encoded into a probe system \(S\), also of arbitrary dimension, through a global unitary interaction \(U_{SE}\) between the probe and the environment. Here, we consider three types of encoding processes. In the CPTP encoding, the probe is initially prepared in a state \(\rho_S\) that is independent of \(T\). In contrast, for NCPTP encodings, the initial probe state is correlated with the environment. We further classify the NCPTP encodings into two types:
\begin{enumerate}
    \item Type-I NCPTP encoding: the initial probe-environment joint state \(\rho_{SE}\) is a pure entangled state.
    \item Type-II NCPTP encoding: the initial probe-environment joint state \(\rho_{SE}\) is an arbitrary correlated state.
\end{enumerate}
The probe-environment initial states in both types of NCPTP encodings are chosen such that the reduced state of the environment is the thermal state $\tau_E$. The encoded probe state \(\rho_S(T)\) is then obtained via the action of the quantum map in Eq.~\eqref{qmap-cptp} for the CPTP encoding, and the quantum map in Eq.~\eqref{qmap-gen} for the two NCPTP encodings.

From metrological considerations, the precision in the estimation is proportional to the QFI of $\rho_S(T)$.  Previous studies demonstrated thermometric advantage in non-equilibrium probe states in comparison with the equilibrium (thermal) state of probe using CPTP encoding~\cite{Sekatski2022optimal,PhysRevA.108.062421,e26070568,PhysRevE.110.024132,trombetti2025nonequilibriumthermometryensembleinitially,prs2-jcxj,saha2026precisionenhancementtransientquantum}, our objective here is to check whether NCPTP encodings can provide quantum Fisher information greater than that of the thermal environment as compared to the CPTP encoding.



The following lemma shows the existence of an upper limit in the thermometric precision under the CPTP encoding.


\begin{lemma}\label{lem1}
Consider an environment $E$ of arbitrary dimension $d_E$, initially prepared in a thermal state $\tau_E$, and a probe system $S$ of arbitrary dimension $d_S$, initialized in a state $\rho_S$. The QFI of the probe, encoded by any CPTP encoding scheme, is upper bounded by that of the thermal state $\tau_E$, i.e.,
\begin{equation}
F_Q\big(\rho_S(T)\big) \leq F_Q(\tau_E).
\end{equation}
\end{lemma}

Here $F_Q\big(\rho_S(T)\big)$ and $F_Q(\tau_E)$ refer to the quantum Fisher information corresponding to the encoded state $\rho_S(T)$ and the thermal state, $\tau_E$ respectively.

\begin{proof}
The initial joint state of the probe and the environment is given by $\rho_S \otimes \tau_E$.
Since $\rho_S$ is independent of the temperature $T$, the additivity of QFI under tensor products implies
\begin{equation}
   F_Q(\rho_S \otimes \tau_E) = F_Q(\rho_S) + F_Q(\tau_E) = F_Q(\tau_E). \label{fish1} 
\end{equation}

For a CPTP encoding, the encoded probe state $\rho_S(T)$ is obtained from Eq.~\eqref{qmap-cptp}, with $\rho_E=\tau_E$. Noting that the partial trace over $E$ is a CPTP map over the composite system $SE$~\cite{lidar2020lecturenotestheoryopen} and that the QFI is monotonically non-increasing under CPTP maps~\cite{Liu_2020}, we have 
\begin{align}
F_Q(\rho_S(T))
&= F_Q\!\left(\operatorname{Tr}_E \!\left[ U_{SE} \left( \rho_S \otimes \tau_E \right) U_{SE}^\dagger \right]\right) \nonumber \\
&\leq F_Q\!\left(U_{SE} \left( \rho_S \otimes \tau_E \right) U_{SE}^\dagger \right). \label{fish2} 
\end{align}

Here, the map $\mathcal{N} = \operatorname{Tr}_E$ is a CPTP map acting on the joint state $U_{SE} \left( \rho_S \otimes \tau_E \right) U_{SE}^\dagger$, in contrast to the CPTP map $\mathcal{N}_{\mathrm{CPTP}}$ that defines the encoding on the system alone.

Further, the QFI is invariant under temperature-independent unitary transformations~\cite{Liu_2020}, which implies
\begin{equation}
    F_Q\!\left(U_{SE} \left( \rho_S \otimes \tau_E \right) U_{SE}^\dagger \right)
= F_Q(\rho_S \otimes \tau_E). \label{fish3} 
\end{equation}

Combining Eqs.~\eqref{fish1}--\eqref{fish3} yields
\begin{equation}
    F_Q(\rho_S(T)) \leq F_Q(\tau_E).
\end{equation}

Thus, under any CPTP encoding, the QFI of the encoded probe state is upper bounded by that of the thermal environment.
\end{proof}
\textcolor{black}{Therefore, the optimal precision obtainable with CPTP maps is limited to the thermal-state value. A crucial question that arises at this point is whether NCPTP maps can enhance the QFI, thereby achieving a precision that is able to surpass this bound. } 

\textcolor{black}{We now examine the effect of the NCPTP encodings on the QFI of the encoded probe state. We begin by considering the type-I NCPTP encoding, for which the initial joint state is taken to be a purification of the thermal state of the environment, with the probe serving as the purifying subsystem. The following proposition establishes an upper bound on the achievable precision in this case.} 
\begin{proposition}
Consider the type-I NCPTP encoding on the probe, realized through an initial probe-environment joint state \(\ket{\Psi}_{SE}=\sum_i \sqrt{p_i}\,\ket{e_i}_S\ket{i}_E\), where the probe and the environment have dimensions \(d_S\) and \(d_E\), respectively, with \(d_S \ge d_E\), such that the reduced state of the environment is the thermal state \(\tau_E=\sum_i p_i \ketbra{i}{i}_E\). Then, the QFI of the encoded probe state \(\rho_S(T)\) is upper bounded by that of \(\tau_E\).
\end{proposition}
\begin{proof}
The initial joint state of the probe and the environment is given by $\rho_{SE}=\ketbra{\Psi_{SE}}{\Psi_{SE}}$. We assume that the only temperature dependence is contained in the coefficients $\{p_i\}$, while the basis vectors $\{\ket{e_i}_S\}$, $\{\ket{i}_E\}$ are independent of $T$.
    By similar arguments as presented in lemma~\ref{lem1}, the QFI of the encoded probe state $\rho_S(T)$ satisfies
\begin{align}
F_Q(\rho_S(T))
&= F_Q\!\left(\operatorname{Tr}_E \!\left[ U_{SE} \rho_{SE} U_{SE}^\dagger \right]\right) \nonumber \\
&\leq F_Q\!\left(U_{SE} \rho_{SE}  U_{SE}^\dagger \right)=F_Q(\rho_{SE}).\nonumber 
\end{align}
Now, the QFI for a pure encoded state $\rho_{SE}$ is given by
\begin{align} F_Q(\rho_{SE})&=4\left(\langle \partial_T\Psi|\partial_T\Psi \rangle_{SE} - |\langle \partial_T\Psi|\Psi \rangle_{SE}|^2\right),\nonumber\\
&=4\left(\sum_{i}\frac{\dot{p}^2_i}{4p_i}-\frac{1}{4}\left|\sum_{i}\dot{p}_i\right|^2\right)=\sum_{i}\frac{\dot{p}^2_i}{p_i},\label{fish_ncp2}
\end{align}
where $\dot{p}_i=dp_i/dT$, and we have used $\textstyle \sum_i \dot p_i = \frac{d}{dT}\textstyle \sum_i p_i = 0$, which follows from the normalization condition $\textstyle \sum_{i}p_i=1$. 

Using Eq.~\eqref{qfi}, the QFI of the state $\tau_E$ is given by
\begin{align}
    F_Q(\tau_E)
&= 2 \sum_{i,j} \frac{\left|\langle i | \partial_T \tau_E | j \rangle\right|^2}{p_i + p_j}= 2 \sum_{i,j,k} \frac{|\dot{p}_k\langle i | k\rangle \langle k | j \rangle|^2}{p_i + p_j},\nonumber\\&=2 \sum_{i,j} \frac{|\dot{p}_i\delta_{ij}|^2}{p_i+p_j}
=2 \sum_{i,j} \frac{\dot{p}^2_i\delta_{ij}}{p_i+p_j}\nonumber
\\&=\sum_{i}\frac{\dot{p}^2_i}{p_i},\label{fish_ncp3} 
\end{align}
where we have used $\partial_T\tau_E = \sum_{k} \dot{p}_k\ketbra{k}{k}$.
Comparing Eqs.~\eqref{fish_ncp2} and \eqref{fish_ncp3} shows that $F_Q(\rho_{SE})=F_Q(\tau_E)$. This establishes that $F_Q(\rho_S(T)) \leq F_Q(\tau_E)$. 
\end{proof}
\color{black}
In practical scenarios, the dimension of the probe is typically less than or equal to that of the environment. We can therefore consider \(d_S=d_E\).

{Since the maximum achievable QFI is the same for the CPTP and type-I NCPTP encodings, we conclude that additional resources do not enhance the metrological sensitivity when the probe-environment joint state is restricted to be pure. Therefore, a resource-efficient CPTP map suffices in this case.}  This motivates us to explore the broadest scenario possible, i.e., arbitrary correlated states $\rho_{SE}$, constrained only by the requirement that the reduced state of the environment be thermal. The resulting encoding corresponds to the type-II NCPTP encoding process. Using the invariance of the QFI under parameter-independent unitary operations, one immediately identifies an upper bound for this case as well:
\begin{equation}
    F_Q(\rho_{S}(T)) \leq F_Q(\rho_{SE}(T)). \label{upp-ncptp-ii}
\end{equation}
However, \(\rho_{SE}\) can be an arbitrary state with a thermal marginal in \(E\), making further analytical simplification difficult.

We therefore turn to a numerical analysis. To this end, we consider both the probe and the environment to be qubits. The local Hamiltonian of each is taken to be the Pauli-\(z\) operator, \(\sigma_z\). Consequently, the total Hamiltonian of the composite system \(SE\) is given by
\begin{equation}
    H_T=\sigma_z \otimes \mathbb{I}_2+\mathbb{I}_2 \otimes \sigma_z. \label{total-hamiltonian}
\end{equation}

Our aim is to compute the optimal QFI under each encoding process for this Hamiltonian. 
Note that the temperature can be expressed in the units of $\hbar\delta/{k_B}$, where $\hbar$ and $k_B$ denote the reduced Planck constant and the Boltzmann constant, respectively, and $\delta$ sets the energy scale. In this work, we set $\hbar=k_B=1$. Throughout the paper, we use the dimensionless temperature $T/{\delta}$, which we continue to denote by $T$ for notational convenience. As a consequence, the QFI is also dimensionless.


\section{Thermometry considering general two-qubit unitaries} \label{gen}
In this section, we compute the optimal QFI for each encoding process for arbitrary two-qubit encoding unitaries. A convenient way to parameterize such unitaries is given by the Kraus-Cirac decomposition~\cite{khaneja2000cartandecompositionsu2nconstructive,PhysRevA.63.062309}
\begin{equation} \label{global-u}
    U_{SE}=U_A \otimes U_B U_d V_A \otimes V_B,
\end{equation}
where $U_A,U_B,V_A, V_B \in SU(2)$. We set $U_A=W_1,U_B=W_2,V_A=W_3,V_B=W_4$, where $W_k \in SU(2), k=1,2,3,4$ have the following form
\begin{equation*}
W_k \coloneqq
\begin{pmatrix}
\cos\frac{\theta_k}{2}\, e^{\tfrac{i}{2}(\psi_k+\nu_k)}
&
\sin\frac{\theta_k}{2}\, e^{-\tfrac{i}{2}(\psi_k-\nu_k)}
\\[6pt]
-\sin\frac{\theta_k}{2}\, e^{\tfrac{i}{2}(\psi_k-\nu_k)}
&
\cos\frac{\theta_k}{2}\, e^{-\tfrac{i}{2}(\psi_k+\nu_k)}
\end{pmatrix}.
\end{equation*}
The range of parameters are given by $\theta_k \in [0,\pi]$, $\nu_k \in [0,2\pi]$ and $\psi_k \in [0,4\pi)$. The nonlocal part $U_d$ can be represented as
\begin{equation}
    U_d=\exp(-i\alpha_x \sigma_x \otimes \sigma_x-i\alpha_y \sigma_y \otimes \sigma_y-i\alpha_z \sigma_z \otimes \sigma_z),
\end{equation}
with $(\sigma_x, \sigma_y)$ being Pauli x- and y-matrices respectively.
The range of the parameters $(\alpha_x,\alpha_y,\alpha_z)$ can be restricted to $[0,\pi/2]$. 

\begin{figure*}
    \centering
    \hspace{-0.5cm}
    \includegraphics[scale=0.45]{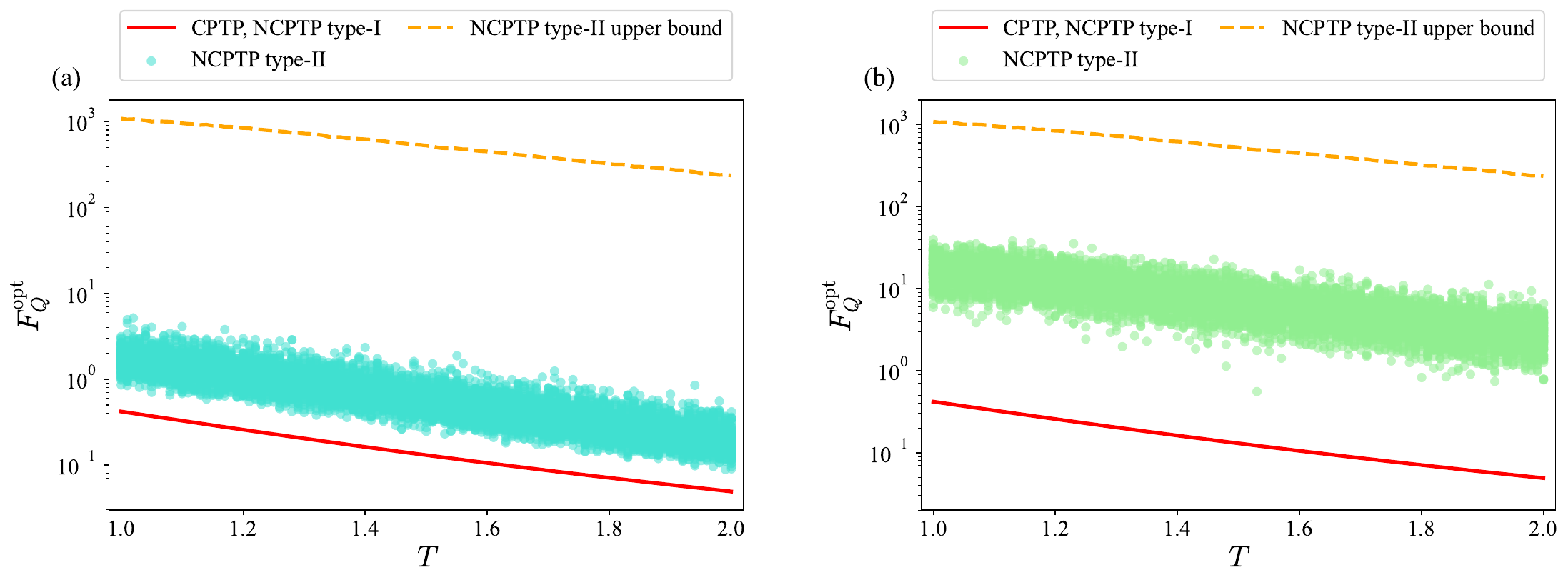}
    \caption{\textbf{Variation of the optimal QFI of the encoded probe state for (a) arbitrary unitaries  and (b) energy-conserving unitaries, with the environment temperature, under the CPTP and the NCPTP encoding processes. } Here, the probe and the environment are both taken as qubits. Across both subplots, the optimal QFI is plotted as a function of the environment temperature $T \in [1.0,2.0]$ in steps of $0.01$. Temperature is plotted along the horizontal axis, whereas the optimal QFI is plotted along the vertical axis, in logarithmic scale. The values for the CPTP and the type-I NCPTP encodings are the same for arbitrary unitaries, as well as energy-preserving unitaries, and equals the QFI of the thermal state of the environment. For each value of $T$, the data shown for the type-II NCPTP encoding are obtained from $100$ separate initial joint states. We observe that the optimal QFI for the type-II NCPTP encoding attain significantly higher values as compared to those for the CPTP and the type-I NCPTP encodings. The plotted upper bound for the NCPTP type-II encoding is computed by averaging over $10$ independent runs of the optimization algorithm. The optimal QFI decreases as $T$ increases for all cases. The axes in all the plots are dimensionless. The colors of curves and markers corresponding to each panel are provided in legend.}
    \label{fig:1}
\end{figure*}

\subsection{Quantum Thermometry under CPTP encoding}\label{cp}
For the CPTP encoding, we consider the state of initial two-qubit probe-environment composite system as $\rho_{SE}=\rho_S \otimes \tau_E$, where $\tau_E$ is the thermal state of the environment given by,
\begin{equation}
    \tau_E \coloneqq \frac{\exp(-\frac{\sigma_z}{T})}{\Tr(\exp(-\frac{\sigma_z}{T}))}=(1-r) \ketbra{0}{0}+r\ketbra{1}{1}. \label{thermal}
\end{equation}

 Here, $r\coloneqq\frac{\exp(1/T)}{\exp(-1/T)+\exp(1/T)}$ and $(\ket{0},\ket{1})$ are the eigenvectors of $\sigma_z$ corresponding to eigenvalues $1$ and $-1$ respectively. The initial probe state can be written in Bloch sphere representation,
\begin{equation}
    \rho_S=\frac{1}{2}(\mathbb{I}_2+n_x\sigma_x+n_y\sigma_y+n_z\sigma_z), \label{bloch}
\end{equation}
where $n_x,n_y,n_z \in \mathbb{R}$ satisfy $n_x^2+n_y^2+n_z^2 \leq 1$. The probe and the environment then interact via the two-qubit global unitary 
$U_{SE}$, described in Eq.~\eqref{global-u}. We subsequently compute the 
quantum Fisher information (QFI) of the encoded probe state 
$\rho_S(T)=\mathcal{N}_{\text{CPTP}}(\rho_S)$ (Eq.~\eqref{qmap-cptp}) and 
maximize it by numerically optimizing over all possible two-qubit unitaries 
$U_{SE}$ and qubit probe states $\rho_S$. The 
derivative $d\rho_S(T)/{dT}$ appearing in Eq.~\eqref{qfi} during the evaluation of the QFI is 
numerically approximated using the five-point midpoint formula, given by
\begin{align}
    \frac{d\rho_S(T)}{dT}
=&
\frac{1}{12h}
\Big[
\rho_S(T-2h)
- 8\,\rho_S(T-h)\nonumber\\&
+ 8\,\rho_S(T+h)
- \rho_S(T+2h)
\Big]. \label{5-pt}
\end{align}
Here, we have considered $h=0.001$. The numerical Optimization is carried out using NLOPT~\cite{NLopt}. We consider the environmental temperature $T$ in the range $[1.0,2.0]$ with a step size of $0.01$. For all values of $T$, we find that the optimal QFI coincides with the QFI of the corresponding thermal state of the environment. 


\subsection{Quantum Thermometry under NCPTP encodings}\label{ncp}

The initial 
joint state for the type-I NCPTP encoding is taken to be the pure entangled state 
$\rho_{SE}=\ketbra{\Psi_{SE}}{\Psi_{SE}}$, with
\begin{equation}
    \ket{\Psi_{SE}}=\sqrt{(1-r)}\ket{00}+\sqrt{r}\ket{11}.
\end{equation}
It can be easily verified that $\Tr_S \rho_{SE}=\tau_E$, which coincides with 
the thermal state considered in the CPTP setting. {We also take into account the class of NCPTP maps that can be 
generated by local unitary rotations of $\ket{\Psi_{SE}}$, while keeping the 
marginal state of the environment thermal. For this purpose, we restrict ourselves 
to local unitaries of the form $U_1 \otimes \mathbb{I}_2$, with $U_1 \in SU(2)$ given by}

\begin{equation*}
U_1
=
\begin{pmatrix}
e^{i\delta}\cos\gamma & e^{i\beta}\sin\gamma \\
-\,e^{-i\beta}\sin\gamma & e^{-i\delta}\cos\gamma
\end{pmatrix},
\end{equation*}
where $\beta, \gamma, \delta \in [0,2\pi]$. The encoded state of the probe is then given by
\begin{equation}  \rho_S(T)=\Tr_E [U_{SE} (U_1 \otimes \mathbb{I}_2) \rho_{SE} (U^\dagger_1 \otimes \mathbb{I}_2) U_{SE}^\dagger], \label{en-con-probe}
\end{equation}
Computing the QFI of this state by optimizing over all choices of $U_{SE}$ and the local unitaries $U_1$ for same set of temperature values as before, we find that the maximum QFI equals $F_Q(\tau_E)$ across all values of temperature and decreases with increasing temperature.

For type-II NCPTP encoding, we need to consider an arbitrary two-qubit state, with the thermal marginal. A general two-qubit density matrix can be decomposed as
\begin{equation}
\rho_{SE}
=
\frac{1}{4}\Bigl(
\mathbb{I}_2\otimes\mathbb{I}_2
+
\vec{a}\cdot\vec{\sigma}\otimes\mathbb{I}_2
+
\mathbb{I}_2\otimes\vec{b}\cdot\vec{\sigma}
+
\sum_{i,j=1}^3 c_{ij}\,\sigma_i\otimes\sigma_j
\Bigr), \label{gen-in-state}
\end{equation}
where all components of $\vec{a}, \vec{b}$ and $C=\{c_{ij}\}$ lie within the interval $[-1,1]$. They must be chosen suitably so as to make $\rho_{SE}$ positive semi-definite, as required for any valid density matrix. 

In our case, we need to find the conditions which makes the reduced environment state a thermal state. Tracing out the probe degree of freedom $\rho_{SE}$ gives us 
\begin{align*}
    \rho_E&\coloneqq\Tr_S\rho_{SE}=\frac{1+b_z}{2}\ketbra{0}{0}+\frac{b_x-ib_y}{2}\ketbra{0}{1}\\&+\frac{b_x+ib_y}{2}\ketbra{1}{0}+\frac{1-b_z}{2}\ketbra{1}{1},
\end{align*}
where $\vec{b}=(b_x,b_y,b_z)$. Comparing this with the form of the thermal state $\tau_E$, we find that $\rho_E=\tau_E$ for the following choice of parameters: $b_x=b_y=0$ and $b_z=1-2r$. Since $r$ depends only on the eigenvalues of $\sigma_z$ and temperature, $b_z$ is also fixed for a given $T$. The remaining 12 parameters may, in principle, be functions of \(T\), although determining their functional dependence is not straightforward. Therefore, for the sake of simplicity, we restrict our attention to a subset of two-qubit states \(\rho_{SE}\) for which these parameters are independent of \(T\). Thus, we need to optimize over $12$ parameters of the state, rather than $15$, in addition to the parameters of the global unitary.

Using these constraints on $\rho_{SE}$, we now compute the QFI of the encoded probe state. For each value of $T$, we show the existence of $100$ example initial joint states for which the QFI exceeds that corresponding to the CPTP and type-I NCPTP encoding. This implies that an enhancement in QFI can be achieved by enlarging the set of initial states to include arbitrary two-qubit correlated states.

The results for all three types of encodings are shown in Fig.~\ref{fig:1}(a). The horizontal and vertical axes depict $T$ and the optimal QFI, $F^\text{opt}_Q$ respectively, with the vertical axis plotted in logarithmic scale for improved visibility. The temperatures are selected from the interval $[1.0,2.0]$, in steps of $0.01$. The plot shows that an optimal QFI beyond $F_Q(\tau_E)$ can be obtained under the NCPTP type-II encoding. Also, we numerically optimize the upper bound of the type-II NCPTP encoding (Eq.~\eqref{upp-ncptp-ii}), over all possible states given by Eq.~\eqref{gen-in-state}, with $b_x=b_y=0$ and $b_z=1-2r$, and plot the value averaged over $10$ independent runs of the optimization algorithm. Since we restrict our analysis to a subclass of the full family of states \(\rho_{SE}(T)\), with the temperature dependence entering only through the coefficient \(b_z\), the upper bound is not reached in this case, although the values are still substantially larger compared to $F_Q(\tau_E)$.
For all cases, optimal QFI decreases with increasing $T$.

\section{Thermometry using two-qubit energy-preserving unitaries}\label{en}
So far, we have considered the full set of two-qubit unitaries acting on the composite probe-environment system for encoding the probe state. Lemma~\ref{lem1} immediately implies that, when the probe and environment have the same dimension, the SWAP gate \(U_{\mathrm{SWAP}}\) is one of the optimal unitaries for the CPTP encoding. For qubit probe and environment, \(U_{\mathrm{SWAP}}\) commutes with the total Hamiltonian \(H_T\) in Eq.~\eqref{total-hamiltonian}, and is therefore energy conserving\footnote{A unitary $U^{C}_{SE}$ is energy-conserving if it commutes with the total Hamiltonian $H_T$ of the probe-environment composite, i.e., $[U^C_{SE},H_T]=0$.}. This motivates us to examine the subclass of energy-conserving unitaries, which are physically well motivated in quantum thermodynamics, where they are used to define thermal operations~\cite{Janzing2000,rf,Horodecki_2013,PhysRevLett.111.250404}, the free operations in the resource-theoretic formulation of quantum thermodynamics. We then ask whether the type 1 NCPTP map can provide any advantage in temperature estimation of the environment or the situation remains same as that of the arbitrary two-qubit unitary case (Where maps are sampled from arbitrary two qubit unitaries) and whether the metrological advantage of the type-II NCPTP encoding persists under this restriction.

 The energy-preserving condition implies that \(U^C_{SE}\) is diagonal in the eigenbasis of \(H_T\) when \(H_T\) has a nondegenerate spectrum, and block diagonal when degeneracies are present.
 For our choice of Hamiltonian $H_T$, its eigenbasis is simply the two-qubit computational basis $\{\ket{00},\ket{01},\ket{10},\ket{11}\}$, with $\ket{01}$ and $\ket{10}$ forming a twofold degenerate subspace corresponding to zero energy. Consequently, 
$U^C_{SE}$ can be decomposed as
\begin{align}
     U^C_{SE}&=e^{-i\lambda_1}\ketbra{00}{00}+e^{-i\lambda_2}\ketbra{\Phi_1}{\Phi_1}\nonumber\\&+e^{-i\lambda_3}\ketbra{\Phi_2}{\Phi_2}+e^{-i\lambda_4}\ketbra{11}{11}, \label{en-con-u}
\end{align}
where $e^{-i\lambda_1},e^{-i\lambda_2},e^{-i\lambda_3},e^{-i\lambda_4}$ are the phase factors, with $\lambda_1,\lambda_2,\lambda_3,\lambda_4\in [0,2\pi]$, and $\{\ket{\Phi_1},\ket{\Phi_2}\}$ is an arbitrary orthonormal basis spanning the degenerate subspace, given by
\begin{align}
    &\ket{\Phi_1}=\sin \lambda_5 \ket{10}+\cos \lambda_5 \ket{01},\nonumber\\&\ket{\Phi_2}=\cos\lambda_5 \ket{10}-\sin \lambda_5 \ket{01}, \label{degen}
\end{align}
 with $\lambda_5 \in [0,2\pi]$. We consider the three encoding schemes in this setting and optimize the QFI using the procedure discussed in~\ref{gen}, restricting the optimization to the set of energy-preserving unitaries $U^{C}_{SE}$.
 The optimal global unitary for the CPTP encoding is numerically found to have the following form:
\begin{equation}
  U^{C,\text{opt}}_{SE}=  \begin{pmatrix}
        e^{ia}& 0& 0& 0\\
        0& 0& e^{ib}& 0\\
        0& e^{ib}& 0& 0\\
        0& 0& 0& e^{ic}\\
    \end{pmatrix}, \label{opt-u-cp}
\end{equation}
where $a,b,c \in [0,2\pi]$. It is readily seen that $U^{C, \text{opt}}_{SE}$ is equivalent to $U_\mathrm{SWAP}$ up to phases. Plugging Eqs.~\eqref{opt-u-cp} and~\eqref{bloch} in Eq.~\eqref{qmap-cptp} gives us the form of the encoded probe state:
\begin{equation}
        \rho_S(T)=(1-r)\ket{0}\bra{0}+r\ket{1}\bra{1}.\nonumber
\end{equation}

Thus, the encoded probe state coincides exactly with the thermal state \(\tau_E\), and therefore the optimal QFI is \(F^\text{opt}_Q = F_Q(\tau_E)\).


Under type-I NCPTP encoding, the optimal probe unitary $U_1$ is numerically observed to be diagonal:
\begin{equation}
    U^\text{opt}_1=\begin{pmatrix}
        e^{ia} &0\\
        0& e^{-ia}
    \end{pmatrix}.
\end{equation}
Note that unlike~\ref{ncp}, local unitaries cannot be absorbed into the energy-conserving global unitaries involved in optimization, since $U_{SE} (U_1 \otimes \mathbb{I}_2)$ may not commute with $H_T$. 
The encoded state of the probe $\rho_S(T)$ is calculated by substituting the form of $U_1$ in Eq.~\eqref{en-con-probe}, and is found to be the thermal state (Eq.~\eqref{thermal}). 



Moving on to the type-II NCPTP encoding, we compute the optimal QFI for $100$ initial joint states, sampled randomly, throughout the temperature range. The results are depicted in Fig~\ref{fig:1}(b). The optimal QFI $F^\text{opt}_Q$ and the environment temperature $T\in [1.0,2.0]$ are plotted along the vertical and the horizontal axes respectively. The vertical axis is plotted in logarithmic scale as before. The temperatures are varied with a step size of $0.01$. We find that at any given temperature, optimal QFI values for each of the $100$ initial joint states corresponding to the type-II NCPTP encoding exceed that of the CPTP and type-I NCPTP encodings by a considerable margin. Consequently, the optimal QFI attainable under different encoding schemes exhibits a similar qualitative behavior for general unitaries, as well as energy-preserving unitaries.  

Further, the optimal values are higher than those obtained by optimizing over arbitrary two-qubit unitaries. This can be understood from the fact that energy-conserving unitaries form a small subset of the full set of two-qubit unitaries. Therefore, the probability of the former getting chosen while random sampling over the latter is negligibly small. The energy-conserving unitaries, however, can yield greater values of the QFI. Hence, restricting the sampling to this subset leads to higher optimal QFI.


As before, the optimized upper bound on the QFI attainable via the type-II NCPTP encoding is obtained by averaging over 10 independent optimization runs, and is plotted in the figure. The upper bound is not attained in this case either. Moreover, increasing $T$ decreases the optimal QFI.


\section{Conclusion} \label{conc}

In this work, we went beyond the conventional framework of quantum thermometry based on completely positive (CP) maps and analyzed how the introduction of non-completely positive (NCP) encoding can affect the precision of temperature estimation. In particular, we addressed the question of whether NCP encoding can outperform CP encoding and yield improved precision.

We answered this in the affirmative. Specifically, we investigated temperature estimation for an arbitrary-dimensional thermal environment under three types of encoding processes: completely positive trace-preserving (CPTP) maps and two classes of physically realizable non-completely positive trace-preserving (NCPTP) maps. The CPTP encoding is realized through initially uncorrelated probe–environment states, whereas type-I and type-II NCPTP encodings correspond to pure entangled and general correlated initial states, respectively. In all cases, the reduced state of the environment was constrained to be thermal, state corresponding to the temperature to be estimated.

We analytically established that, for both CPTP and type-I NCPTP encoding, the optimal quantum Fisher information (QFI) is upper bounded by the QFI of the thermal state of the environment. Furthermore, for a qubit probe interacting with a qubit environment, we showed that this bound is saturable. In contrast, under type-II encoding, we demonstrated that this bound can be significantly surpassed. In particular, for qubit probes and environments, we observed enhanced precision beyond the thermal-state limit when optimizing over general two-qubit correlated states and global unitaries. This enhancement persisted even when the set of global unitaries was restricted to physically relevant energy-conserving interactions.

Our results show that CPTP encoding, implemented using uncorrelated initial states, can achieve the same precision as type-I NCPTP encoding based on pure entangled states. However, allowing general correlated initial states, leading to type-II NCPTP encoding, provides a genuine precision advantage. Overall, our findings highlight the role of initial correlations and encoding structure in determining thermometric performance, and demonstrate that going beyond the CPTP paradigm can be beneficial for practical quantum metrology.

\section*{Acknowledgments}
We acknowledge the use of Armadillo and QIClib (\url{https://titaschanda.github.io/QIClib})
PC and DS acknowledge support from the ‘INFOSYS scholarship for senior students’ at Harish-Chandra Research Institute, India.

\appendix
\section{Illustrative Examples}
\begin{figure*}
    \centering
    \includegraphics[scale=0.48]{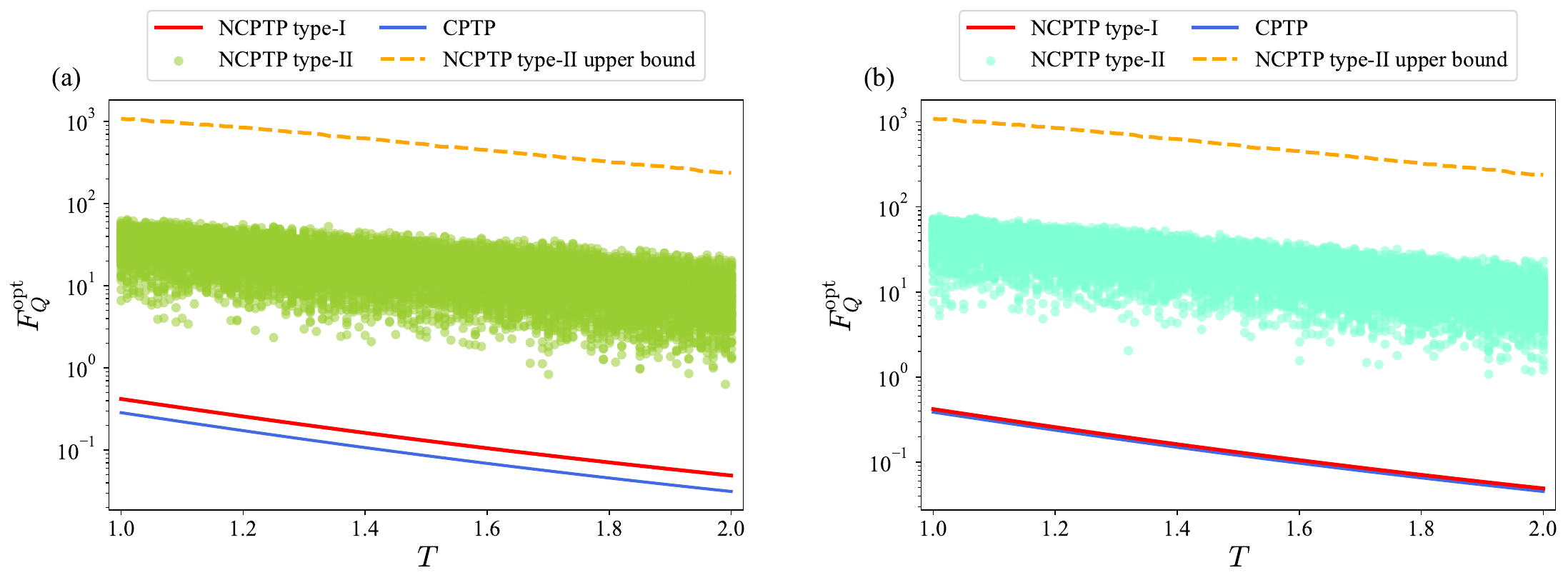}
    \caption{\textbf{Behaviour of optimal QFI of the encoded probe state, for global unitary generated via (a) two-qubit XX interaction  and (b) two-qubit anisotropic  XY interaction, with the environment temperature, corresponding to the CPTP and the NCPTP encodings. } The main plots in both figures represent the variation of the optimal QFI with environment temperature $T \in [1.0,2.0]$, with a step size of $0.01$. The horizontal axis denotes the temperature, while the vertical axis denotes the corresponding optimal QFI, on a logarithmic scale. The CPTP-optimal QFI values are higher in panel \textbf{(b)} than panel \textbf{(a)}; however, they remain below the type-I NCPTP optimum values in both panels, which equals the thermal-state QFI. For the type-II NCPTP encoding, the optimal QFI (corresponding to $100$ joint initial states for each $T$) reaches appreciably large values. The upper bound on the optimal QFI for type-II NCPTP encoding is plotted by averaging over $10$ independent runs of the optimization algorithm. An increase in the temperature reduces the optimal QFI, as observed for all cases. All parameters in the plots are dimensionless. The colors of curves and markers corresponding to each panel are provided in legend.}
    \label{fig:2} 
\end{figure*}
Here we present our analysis through examples of specific energy-preserving and energy non-preserving unitaries using a physical model. We consider a fixed global unitary $ U_{SE}=\exp[-i(H_T+H_{int})]$, where
\begin{align}
   H_{int}=J_x\sigma_x \otimes \sigma_x+J_y \sigma_y \otimes \sigma_y
\end{align}
describes the antiferromagnetic two-qubit XY model~\cite{LIEB1961407}. In general, $J_x \neq J_y$ , which corresponds to the anisotropic XY model. For $J_x=J_y=J$, one obtains the isotropic XY (XX) model, which can equivalently be written as
\begin{align}
   H_{int}=2J(\sigma_+\otimes\sigma_-+\sigma_-\otimes\sigma_+),
\end{align}
with $\sigma_{\pm}=\frac{\sigma_x\pm i\sigma_y}{2}$. 



Here we consider (a) the XX model with $2J=1$ units, and (b) the anisotropic XY model with $J_x=1, J_y=0.5$ units. It can easily be checked that the XX model (a) is energy-preserving with respect to $H_T$, whereas the same is not true for the anisotropic XY model (b). We compute the maximum QFI by optimizing over probe states for CPTP encoding, and over local unitaries acting on the probe, and the general two-qubit joint states for NCPTP encodings. For the CPTP encoding, the optimal QFI obtained numerically for both the XX and anisotropic XY models lies below the corresponding thermal-state QFI of the environment, indicating that the unitaries considered in these cases are not optimal for the CPTP encoding. However, the optimal QFI is higher for the anisotropic XY model than for the XX model. For type-I NCPTP encoding, both models give optimal QFI saturating the thermal state $\tau_E$. Lastly, the optimal QFI surpasses $F_Q(\tau_E)$ for both models, under the type-II NCPTP encoding. Thus, in these models, both types of NCPTP encodings are able to provide thermometric advantage over the CPTP encoding.

The results are presented in Fig.~\ref{fig:2}. The optimal QFI $F^\text{opt}_Q$ (vertical axis) is plotted as a function of the environment temperature $T$ (horizontal axis). The vertical axis is plotted using logarithmic scale. The values of $T$ are chosen from the interval $[1.0,2.0]$ with a spacing of $0.01$. Panels (a) and (b) respectively represent the XX and anisotropic XY models. For NCPTP type-II encoding, $100$ initial states $\rho_{SE}$ are considered for each $T$, and the optimal QFI values for all of them is found to exceed the corresponding CPTP and NCPTP type-I values. Furthermore, the optimal QFI values are spread over a wider range of values in the anisotropic XY model. In both panels, the optimal QFI falls as the temperature rises. We also plot the optimized upper bound on the QFI achievable under the type-II NCPTP encoding, averaged over 10 independent runs. The higher optimal QFI values obtained for the type-II NCPTP encoding in this case, relative to those achievable under general two-qubit unitaries or even general energy-preserving two-qubit unitaries, may be attributed to more effective sampling of the initial states, which in turn leads to higher estimation precision.

\bibliography{ref}
\end{document}